\documentclass{article}
\usepackage[utf8]{inputenc}
\usepackage{amsmath}
\usepackage{amssymb}
\usepackage{amsthm}
\usepackage{mathtools}
\usepackage[hyphens]{url}
\usepackage{hyperref}
\usepackage[hyphenbreaks]{breakurl}
\usepackage{booktabs}
\usepackage{color}
\usepackage{algorithm,algpseudocode}
\usepackage{geometry}

\newcommand{\F}{\mathbb{F}}
\newcommand{\AutLE}{\mathsf{Aut_{LE}}}
\newcommand{\GL}{\mathrm{GL}}
\newcommand{\AGL}{\mathrm{AGL}}

\newcommand{\ord}{\mathrm{ord}}

\newcommand{\rcf}{\mathrm{RCF}}
\newcommand{\companion}{\mathrm{Comp}}
\newcommand{\wt}{\mathrm{wt}}
\newcommand{\diag}{\mathrm{diag}}

\newtheorem{definition}{Definition}
\newtheorem{lemma}{Lemma}
\newtheorem{proposition}{Proposition}
\newtheorem{theorem}{Theorem}

\title{New Instances of Quadratic APN Functions\thanks{This  work  was  funded  by Deutsche  Forschungsgemeinschaft  (DFG);  project  number 411879806 and by DFG under Germany's Excellence Strategy - EXC 2092 CASA - 390781972.}\thanks{This is the version accepted to IEEE Transactions on Information Theory. DOI of the final published version: \href{https://dx.doi.org/10.1109/TIT.2021.3120698}{10.1109/TIT.2021.3120698}.}\thanks{
$\copyright$ 2021 IEEE.  Personal use of this material is permitted.  Permission from IEEE must be obtained for all other uses, in any current or future media, including reprinting/republishing this material for advertising or promotional purposes, creating new collective works, for resale or redistribution to servers or lists, or reuse of any copyrighted component of this work in other works.}}
\renewcommand\footnotemark{}
\author{Christof Beierle and Gregor Leander \\ \\
Ruhr University Bochum, Bochum, Germany}
\date{}

\begin{document}

\maketitle

\begin{abstract}
    In a recent work, Beierle, Brinkmann and Leander presented a recursive tree search for finding APN permutations with linear self-equivalences in small dimensions. In this paper, we describe how this search can be adapted  to find many new instances of quadratic APN functions. In particular, we found 12,921 new quadratic APN functions in dimension eight, 35 new quadratic APN functions in dimension nine and five new quadratic APN functions in dimension ten up to CCZ-equivalence. Remarkably, two of the 35 new APN functions in dimension nine are APN permutations. 
    
    Among the 8-bit APN functions, there are three extended Walsh spectra that do not correspond to any of the previously-known quadratic 8-bit APN functions and, surprisingly, there exist at least four CCZ-inequivalent 8-bit APN functions with linearity $2^7$, i.e., the highest possible non-trivial linearity for quadratic functions in dimension eight.  
    
    {\bf Keywords:} almost perfect nonlinear, Walsh spectrum, linearity, self-equivalence, EA-equivalence
\end{abstract}

\section{Introduction}
Vectorial Boolean functions are used as S-boxes in many block ciphers and thus belong to the fundamental building blocks in symmetric cryptography. When such functions are used in actual cryptographic designs, one has to ensure that they fulfill certain criteria in order to prevent cryptanalytic attacks.  Functions offering the best possible resistance against differential attacks~\cite{DBLP:journals/joc/BihamS91} are called \emph{almost perfect nonlinear (APN)}. 

\begin{definition}{\cite{DBLP:conf/crypto/NybergK92}}
Let $n$ be a positive integer. A function $F\colon \F_2^n \rightarrow \F_2^n$ is called \emph{almost perfect nonlinear} (APN) if, for every $a \in \F_2^n\setminus \{0\} ,b \in \F_2^n$, the equation $F(x) + F(x+a) = b$ has at most 2 solutions $x \in \F_2^n$.
\end{definition}

We know several infinite families of APN functions as well as some sporadic instances. The vast majority of the known instances corresponds to monomial functions over the finite field $\F_{2^n}$ or to functions with algebraic degree 2 over $\F_2^n$ (i.e., quadratic functions). Indeed, at the time of writing, only a single APN instance is known that is not CCZ-equivalent to a monomial function or to a quadratic function~\cite{DBLP:journals/amco/EdelP09}. The \emph{linearity} is a measure for how well a function can be approximated by an affine function and is in particular of importance for resistance of block ciphers against linear cryptanalysis~\cite{DBLP:conf/eurocrypt/Matsui93}. Interestingly, APN functions often have low linearity as well.  While this is always true for quadratic APN functions in odd dimension $n$ (since those functions are almost bent~\cite{ccz}), the situation is more complicated for even values of $n$. For instance, prior to our work, two possible values of the linearity and three possible extended Walsh spectra of quadratic APN functions in dimension $n=8$ were known~\cite[Table 4.2]{arshad}. Moreover, in dimension $n=6$, we know one quadratic APN instance that admits the highest possible linearity of $2^{n-1}$, see the list in~\cite{DBLP:journals/amco/EdelP09}. It is an open question whether such a high linearity is also achievable by quadratic APN functions in higher dimensions~\cite{DBLP:journals/tit/Carlet18}. 

A big open problem is to find new instances of APN \emph{permutations} in even dimensions. Until now, only a single instance of an APN permutation in even dimension is known, namely for $n=6$~\cite{browning2010apn}.  
It is well known that a quadratic APN function in even dimension cannot be a permutation~\cite{DBLP:conf/eurocrypt/SeberryZZ94}. However, the aforementioned permutation is CCZ-equivalent to a quadratic function. The potential to discover new APN permutations in even dimension that are CCZ-equivalent to quadratic functions  is one motivation to explicitly search for quadratic APN functions.

For $n \leq 5$, a complete classification of APN functions up to CCZ-equivalence is known~\cite{DBLP:journals/dcc/BrinkmannL08} and for $n=6$, such a classification is known for APN functions up to algebraic degree three~\cite{langevin}. In 2009, Edel and Pott introduced the switching construction and found new APN functions in dimensions  $n\leq8$ by replacing components of previously-known APN functions~\cite{DBLP:journals/amco/EdelP09}. This led to the discovery of new APN functions, bringing up the number of known CCZ-inequivalent APN functions in dimension $n=7$ and $n=8$ to 19 and 23, respectively. A breakthrough was achieved by the works~\cite{weng2013quadratic} and~\cite{DBLP:journals/dcc/YuWL14}, in which the authors found many new quadratic APN instances in dimension $n=7$ and $n=8$.  Recently, a new quadratic APN function in dimension 7 was found and a complete classification of quadratic APN functions in dimension 7 was achieved~\cite{cryptoeprint:2020:1113,cryptoeprint:2020:1515}. Excluding our results, there are 491 known APN instances in dimension $n=7$ and 8,192 known APN instances in dimension $n=8$ at the time of submission of this manuscript\footnote{In 2021, after the preprint of our work was made public, Yu and Perrin~\cite{cryptoeprint:2021:574} found more than 5,400 other new quadratic APN instances in dimension eight by the same method as in~\cite{DBLP:journals/dcc/YuWL14}, bringing up the total number of known 8-bit APN instances (including our results) to over 26,500. All of the APN instances found in~\cite{cryptoeprint:2021:574} have an extended Walsh spectrum that was already known prior to our work (i.e., $\mathcal{W}_0, \mathcal{W}_1$, or $\mathcal{W}_2$ as defined in Section~\ref{sec:search_le_aut}).} in December 2020. Besides searching for APN instances in fixed dimensions, several infinite families of APN functions have been found, see~\cite{DBLP:journals/ffa/BudaghyanCV20} for a recent summary.

In~\cite{beierle2020}, the authors utilized a recursive tree search for finding APN permutations with linear self-equivalences and were able to classify all APN permutations with linear self-equivalences in dimension $n=6$. 

\subsection{Our Contribution}
By adapting the algorithm of~\cite{beierle2020}, we find many new instances of quadratic APN  functions. The search strategy is conceptually very simple. The basic idea is to fix the look-up table of the APN function $F$ 
entry by entry. Each time a new entry is fixed, besides checking whether there is a contradiction to the property of being APN, the algorithm further checks whether the values in the look-up table already imply the existence of a monomial of algebraic degree higher than two in the algebraic normal form of $F$. The main difference to the algorithm of~\cite{beierle2020} is the incorporation of the check for the existence of high-degree monomials and the removal of the restriction on the bijectivity of the output function. 

By using this approach, also combined with considering linear self-equivalences, we find 12,733 new CCZ-inequivalent quadratic APN functions in dimension $n=8$, which is a substantial increase compared to the 8,192 previously-known APN instances in dimension $n=8$. 

Searching for new instances of APN functions in dimensions higher than $n=8$ is known to be very hard and resource consuming.  Indeed, to the best of our knowledge, previous search methods have not been very successful in finding new APN instances and the only previously-known APN instances in  dimension $n\in\{9,10\}$ are either monomial functions, one of the polynomials with coefficients in $\F_{2}$~\cite{cryptoeprint:2019:1491}, or those that come from the infinite families given in~\cite{budaghyan2009construction},~\cite{DBLP:journals/tit/BudaghyanCCCV20},~\cite{budaghyan2020generalized},~\cite{DBLP:journals/tit/BudaghyanC08}, \cite{DBLP:journals/ffa/BudaghyanCL09}, \cite{DBLP:journals/dcc/Taniguchi19a}, or \cite{DBLP:journals/iacr/BudaghyanHK19}.  Clearly, our approach becomes less efficient as well. However, we are still able to present 35 new APN instances in dimension $n=9$ and five new APN instances in dimension $n=10$. Remarkably, two of the new 9-bit APN functions are permutations. Until now, the only known APN permutations up to CCZ-equivalence were the monomial functions in odd dimension $n$, the binomial family for $3 \mid n$  presented in~\cite{DBLP:journals/tit/BudaghyanCL08} and the sporadic 6-bit APN permutation found in~\cite{browning2010apn}. Up to EA-equivalence, the two new APN permutations can be given in univariate representation over $\F_{2^9}$ by
\begin{align*}
   x \mapsto \ &x^3 + ux^{10} + u^2x^{17} + u^4x^{80} + u^5x^{192},\\ %\\
   x \mapsto \ &x^3 + u^2x^{10} + ux^{24} + u^4x^{80}+ u^6x^{136},
\end{align*}
where $u \in \F_{2^9}^*$ is an element with minimal polynomial $X^3+X+1 \in \F_2[X]$.

Among the extended Walsh spectra of the APN functions found for $n=8$, there are three that do not correspond to any of the previously-known quadratic 8-bit APN functions. In particular, there are four pairwise CCZ-inequivalent  APN functions in dimension 8 having linearity $2^{7}$. 
One such example is the APN function
\begin{align*}
   x \mapsto \ &x^{3} + g^{60} x^{5} + g^{191} x^{6} + g^{198} x^{9} + g^{232} x^{10} + g^{120} x^{12} + g^{54} x^{17} + g^{64} x^{18} + g^{159} x^{20} + \\
   &g^{144} x^{24} + g^{248} x^{33} + g^{203} x^{34} + g^{32} x^{36} + g^{18} x^{40} + g^{216} x^{48} + g^{78} x^{65} + g^{46} x^{66} + g^{91} x^{68} + \\
   &g^{27} x^{72} + g^{70} x^{80} + g^{52} x^{96} + g^{224} x^{129} + g^{18} x^{130} + g^{197} x^{136} + g^{253} x^{144} + x^{160} 
\end{align*}
over $\F_{2^8}$, where $g \in \F_{2^8}^*$ is an element with minimal polynomial $X^8+X^4+X^3+X^2+1 \in \F_2[X]$.

Finally, we apply the switching construction of Edel and Pott~\cite{DBLP:journals/amco/EdelP09} to all the known (by the time of submission of this manuscript) and new quadratic APN instances in dimension $n=7$ and $n=8$, which leads to the discovery of another 188 CCZ-inequivalent APN instances in dimension $n=8$. By using sboxU~\cite{Perrin_sboxu}, we have checked that none of the APN functions we found for $n \in \{8,10\}$ are CCZ-equivalent to a permutation.

The source code of our algorithms is publicly available at~\cite{our_code}. The look-up tables of the new APN instances are available in~\cite{dataset}. We emphasize that our search method is non-exhaustive, so we do not make any claim on the completeness of our findings.

\section{Preliminaries}
For a positive integer $n$, let us denote by $\F_{2^n}$ the field with $2^n$ elements and let $\F_2^n$ denote the $n$-dimensional vector space over $\F_2$. Let $\F_{2^n}^*$ denote the set $\F_{2^n} \setminus \{0\}$. By $\GL(n,\F_2)$ we denote the group of invertible $n \times n$ matrices over $\F_2$ and by $\AGL(n,\F_2)$ we denote the group of affine permutations on $\F_2^n$. Any such affine permutation can be represented as $x \mapsto Lx + b$ for $L \in \GL(n,\F_2)$ and $b \in \F_2^n$. To simplify notation, we are going to use elements of $\GL(n,\F_2)$ and the linear functions that they represent interchangeably throughout this work. In other words, for an element $L \in \GL(n,\F_2)$, we denote the linear function $x \mapsto Lx$ by $L$ as well. The symbol $I_n$ denotes the identity matrix in $\GL(n,\F_2)$. By $\diag(M_1,M_2,\dots,M_k)$, we denote the block-diagonal matrix consisting of the $k$ blocks $M_1,\dots,M_k$, where $M_1$ corresponds to the upper-left block.

For an element $M \in \GL(n,\F_2)$, we denote by $\ord(M)$ the \emph{multiplicative order} of $M$, which is defined as the smallest positive integer $i$ such that $M^{i} = I_n$. Similarly, for  $x \in \F_2^n$, we denote by $\ord_M(x)$ the smallest positive integer $i$ for which $M^i(x) = x$. The \emph{minimal polynomial} of a matrix $M$ over $\F_2$ is defined as the polynomial $p \in \F_2[X]$ of least positive degree such that $p(M)=0$.

For a polynomial 
$q = X^n + q_{n-1}X^{n-1} + \dots + q_1X + q_0 \in \mathbb{F}_2[X]$, the \emph{companion matrix} of $q$ is defined as the $n \times n$ matrix
\[
\companion(q) \coloneqq \left[\begin{array}{ccccc}
0 &  & &  & q_0 \\
1 & 0 &  & & q_1 \\
 & \ddots & \ddots &  & \vdots \\
 &  & 1 & 0 & q_{n-2} \\
 &  &  & 1 & q_{n-1}
\end{array}\right],
\]
which is an element of $\GL(n,\F_2)$ if and only if $q_0=1$. For $x \in \F_2^n$, we denote by $\wt(x)$ the Hamming weight of $x$, which is defined as the number of non-zero coordinates of $x$.

\subsection{Representations of Vectorial Boolean Functions}
For a comprehensive introduction to (vectorial) Boolean functions, we refer to~\cite{carlet_2021}. Here, we recall the most important concepts needed in the remainder of the paper. Note that in the following, we restrict to the case of functions from $\F_2^n$ to itself.

A vectorial Boolean function $F \colon \F_2^n \rightarrow \F_2^n$ can be uniquely expressed as a multivariate polynomial in $\F_2^n[X_1,\dots,X_n]/(X_1^2+X_1,\dots,X_n^2+X_n)$, called the \emph{algebraic normal form (ANF)}. In particular, there exist $a_u \in \F_2^n$ such that 
\[ F(x_1,\dots,x_n) = \sum_{u \in \F_2^n} \left( a_u \prod_{i \in \{1,\dots,n\}}x_i^{u_i}\right).\]
The \emph{algebraic degree} of $F$ is defined as $\max\{\wt(u) \mid a_u \neq 0, u \in \F_2^n\}$. The function $F$ is called \emph{affine} if it is of algebraic degree at most 1 and it is called \emph{quadratic} if it is of algebraic degree 2. The coefficients $a_u$ of the ANF can be obtained by the so-called binary Möbius transform via
\begin{eqnarray}\label{eqn:moebius}
a_u &=& \sum_{x \in \F_2^n, x \preceq u} F(x),
\end{eqnarray} 
where the relation $x \preceq u$ holds if and only if, for all $i \in \{1,\dots,n\}$, we have $(u_i=0 \Rightarrow x_i = 0)$.

Moreover, we can uniquely represent any vectorial Boolean function $F \colon \F_2^n \rightarrow \F_2^n$ as a function from $\F_{2^n}$ to itself via $x \mapsto f(x)$ with $f \in \F_{2^n}[X]/(X
^{2^n}+X)$. This is called the \emph{univariate representation} of $F$.

The \emph{Walsh transform} of $F$ at $(\alpha,\beta) \in \F_2^n \times \F_2^n$ is defined as the sum
\[ \widehat{F}(\alpha,\beta) =  \sum_{x \in \F_2^n} (-1)^{\langle \alpha,x \rangle + \langle \beta,F(x)\rangle}\]
over the integers,
where $\langle x,y \rangle$ denotes the inner product of the vectors $x,y \in \F_2^n$, defined as $\langle x,y \rangle \coloneqq \sum_{i=1}^{n} x_iy_i \mod 2$. The multiset $\{ |\widehat{F}(\alpha,\beta)| \mid \alpha,\beta \in \F_2^n\}$ is called the \emph{extended Walsh spectrum} of $F$. For $\beta \in \F_2^n \setminus \{0\}$, the function $F_{\beta}\colon x \mapsto \langle \beta, F(x)\rangle$ is called a \emph{component} of $F$ and the value $\max_{\alpha \in \F_2^n, \beta \in \F_2^n \setminus \{0\}} |\widehat{F}(\alpha,\beta)|$ is called the \emph{linearity} of $F$. The linearity of $F$ can be understood as a measure of how well a component of $F$ can be approximated by an affine function. In particular, a linearity of $2
^n$ corresponds to the case of $F$ having an affine component.

The \emph{difference distribution table (DDT)} of a function $F\colon \F_2^n \rightarrow \F_2^n$ is the $2^n \times 2^n$ integer matrix (where the rows and the columns are indexed by $\alpha \in \F_2^n$ and $\beta \in \F_2^n$, respectively) that contains $|\{x \in \F_2^n \mid F(x) + F(x+ \alpha) = \beta\}|$ in the entry in row $\alpha$ and column $\beta$. The \emph{differential spectrum} of $F$ is defined as the multiset of entries in the DDT of $F$.

\subsection{Equivalence Relations of Vectorial Boolean Functions}
Let $F, G \colon \F_2^n \rightarrow \F_2^n$. There are several well-known equivalence relations on vectorial Boolean functions that preserve both the differential spectrum and the extended Walsh spectrum. The function $G$ is \emph{linear-equivalent} to $F$ if there exist $A,B \in \GL(n,\F_2)$ such that $F \circ A = B \circ G$.  Moreover, $G$ is \emph{extended affine-equivalent} (\emph{EA-equivalent}) to $F$ if there exist $A,B \in \AGL(n,\F_2)$ and an affine, not necessarily invertible, function $C\colon \F_2
^{n} \mapsto \F_2^n$ such that $F \circ A = B \circ G + C$. The functions $F$ and $G$ are called \emph{CCZ-equivalent}~\cite{DBLP:journals/tit/BudaghyanCP06,ccz} if there exists a transformation $\sigma \in \AGL(2n,\F_2)$ such that $\Gamma_G = \sigma(\Gamma_F)$, where $\Gamma_F \coloneqq \{ (x,F(x)) \mid x \in \F_2^n\}$ denotes the graph of $F$. Among the notions of equivalence listed above, CCZ-equivalence is the most general.  An important goal in the study of APN functions is to classify them up to CCZ-equivalence and the following is a useful result for quadratic APN functions.

\begin{theorem}[\cite{yoshiara2012equivalences}]
Two quadratic APN functions $F, G \colon \F_2^n \rightarrow \F_2^n$ are EA-equivalent if and only if they are CCZ-equivalent.
\end{theorem}

Therefore, since our focus is on quadratic APN functions, we are going to separate the functions we find up to EA-equivalence (which corresponds to CCZ-equivalence in this case) and we only provide a single representative from each CCZ-equivalence class. Such a CCZ-equivalence class will also be called an \emph{instance} throughout this paper, which is represented by one member of the class.  

The \emph{LE-automorphism} group (see~\cite{beierle2020}) of a function $F \colon \F_2^n \rightarrow \F_2^n$ is defined as
\[ \AutLE(F) \coloneqq \left\{ \diag(A,B) \in \GL(2n,\F_2) \mid  A,B \in \GL(n,\F_2) \text{ and } F\circ A = B \circ F \right\}.\]

If a function $F$ admits a non-trivial automorphism $\diag(A,B) \in \AutLE(F)$, we also say that $F$ is \emph{linearly self-equivalent} with respect to the tuple $(B,A)$. 
Note that if $F\colon \F_2^n \rightarrow \F_2^n$ and $G\colon \F_2^n \rightarrow \F_2^n$ are linear-equivalent, we have $\AutLE(F) \cong \AutLE(G)$, where $\cong$ denotes the isomorphism relation of groups. 

\section{A Recursive Tree Search for Quadratic APN Functions}
Our idea for finding new instances of APN functions is to apply a recursive tree search very similar to~\cite[Algorithm 1]{beierle2020}. The main differences are that in the algorithm given here we drop the restriction to only search for permutations and we apply an additional filter for skipping branches within the search tree that lead to functions of algebraic degree higher than two. In this section, we explain the most basic variant that searches for arbitrary $n$-bit quadratic APN functions for $n\geq 2$. It is formally specified in Algorithm~\ref{alg:search}. The algorithm can easily be adapted to only search for quadratic APN functions admitting a particular LE-automorphism (see Section~\ref{sec:search_le_aut}). 

The global array $\mathrm{sbox}$ is initialized to be undefined ($\bot$) at each entry. This array corresponds to the look-up table of the APN function $F$ to be constructed. In each iteration of  $\textsc{NextVal}$, the procedure $\textsc{isComplete}$ first checks whether $\mathrm{sbox}$ is already completely defined. If this is the case, the algorithm has found a quadratic APN function and prints $\mathrm{sbox}$ as the solution. Otherwise, the procedure $\textsc{NextFreePosition}$ is called which selects the next undefined entry $x$ and sets $F(x)$ to a value $y$ that is randomly selected from among a predefined list of possible choices. The orders in which those values $y$ are selected at each depth are determined by the $\textsc{Shuffle}$ procedure that is performed in the beginning. After fixing $F(x)$, the procedure $\textsc{AddPoint}$ checks whether $F$ can still be both APN and quadratic. If not, the current branch of the search tree is skipped and $x$ is set to the next possible value $y$. In case no contradiction to either property is encountered, the algorithm goes one level deeper.

Note that, since the running time of Algorithm~\ref{alg:search} can be very long in cases where no quadratic APN function is found, we abort and restart after a predetermined amount of time (e.g., 10 seconds for $n=7$).

\subsection{APN Check} Each time we set $F(x)$ to a new value $y$, we need to check that the APN property of $F$ has not already been violated. This is performed in exactly the same way as in~\cite{beierle2020}. In particular, the function $\textsc{addDDTInformation}{(x)}$ dynamically changes the DDT according to the value set by the current iteration. The DDT is stored in a global array which is initialized to 0 before calling~Algorithm \ref{alg:search}. Similarly, each time we reset $F(x)$ to $\bot$, the function $\textsc{removeDDTInformation}{(x)}$ applies the appropriate changes to the DDT. Note that $\textsc{addDDTInformation}{(x)}$ returns 1 if the APN property is not violated by fixing $F(x)$. Otherwise, it returns 0. Similarly, $\textsc{removeDDTInformation}{(x)}$ returns 1 if the APN property has not been violated by fixing $F(x)$. Otherwise, it returns 0.

\subsection{Algebraic Degree Check} Each time we set $F(x)$ to a new value $y$, we check whether we can deduce the existence of a monomial of algebraic degree higher than $2$ in the algebraic normal form of $F$. For this, using Equation (\ref{eqn:moebius}), we keep track of the partial sums for all $a_u$ with $\wt(u) \geq 3$ in a global array $\mathrm{sum}$ and update them whenever $x\preceq u$. The check is performed by calling the function  $\textsc{addDegreeInformation}{(x)}$, which is defined below. The function dynamically changes the global arrays $\mathrm{ctr}$ and $\mathrm{sum}$, both of size $2^n$, which are initialized to 0 before calling~Algorithm \ref{alg:search}. Below, the symbol $\oplus$ denotes the bitwise XOR operation in order to distinguish it from the addition of integers, denoted $+$.

\begin{algorithmic}[1]
		\Function{addDegreeInformation}{x}
		\For {$u \in [1,\dots,2^n-1]$ such that $\wt(u)\geq 3$ and $x \preceq u$}
		    \State $\mathrm{ctr}[u] \gets \mathrm{ctr}[u]+1$
		    \State $\mathrm{sum}[u] \gets \mathrm{sum}[u] \oplus \mathrm{sbox}[x]$
		    \If {$\mathrm{ctr}[u] = 2^{\wt(u)}$} \Comment{All $x$ with $x \preceq u$ have been considered}
		        \If{$\mathrm{sum}[u] \neq 0$} \Comment{We have $a_u \ne 0$ in the ANF of $F$}
		            \State \Return 0
		        \EndIf
		    \EndIf
		\EndFor
		\State \Return 1
		\EndFunction
\end{algorithmic}

When $F(x)$ is reset to $\bot$, we need to restore the values of the arrays $\mathrm{ctr}$ and $\mathrm{sum}$ by calling the following procedure.

\begin{algorithmic}[1]
		\Function{removeDegreeInformation}{x}
		\For {$u \in [1,\dots,2^n-1]$ such that $\wt(u)\geq 3$ and $x \preceq u$}
		    \State $\mathrm{ctr}[u] \gets \mathrm{ctr}[u]-1$
		    \State $\mathrm{sum}[u] \gets \mathrm{sum}[u] \oplus \mathrm{sbox}[x]$
		    \If {$\mathrm{ctr}[u] = 2^{\wt(u)}-1$}
		        \If{$\mathrm{sum}[u] \neq \mathrm{sbox}[x]$}
		            \State \Return 0
		        \EndIf
		    \EndIf
		\EndFor
		\State \Return 1
		\EndFunction
\end{algorithmic}

\renewcommand{\algorithmicrequire}{\textbf{Input:}}
\renewcommand{\algorithmicensure}{\textbf{Output:}}
\algnewcommand{\algorithmicgoto}{\textbf{go to}}
\algnewcommand{\Goto}[1]{\algorithmicgoto~\ref{#1}}
\begin{algorithm}
	\caption{\textsc{QuadraticAPNSearch}} \label{alg:search}
	\begin{algorithmic}[1]
		\Require Global array $\mathrm{sbox}$ of size $2^n$, initialized to $\mathrm{sbox}[i] = \bot$, for all $i \in \{0,\dots,2^n-1\}$. Global 2-dimensional array $P$ of size $2^n \times 2^n$ with each $P[i]$ initialized to $[0,\dots,2^n-1]$ for all $i \in \{0,\dots,2^n-1\}$.
		\Ensure Prints an $n$-bit APN function $F$ of algebraic degree at most two.
		    \vspace{.1em}
		    \For {$i \in [0,\dots,2^n-1]$}
		        \State $\textsc{Shuffle}(P[i])$ \Comment{Generates a random permutation of $P[i]$}
		    \EndFor
		    \State $\mathrm{sbox}[0] \gets 0$
		    \State $\textsc{addPoint}(0)$
		    \State $\textsc{nextVal}(0)$
            \vspace{1em}
		\Function{nextVal}{depth}
		    \If{$\textsc{isComplete}(\mathrm{sbox})$} \Comment{Checks if $\mathrm{sbox}$ contains no $\bot$}
		        \State Print $\mathrm{sbox}$ and terminate
		    \EndIf
		    \State $x \gets \textsc{nextFreePosition}()$ \Comment{Chooses the smallest $i$ s.t. $\mathrm{sbox}[i] \neq \bot$}
		    \For{$z \in [0,\dots,2^n-1]$}
		                \State $y \gets P[\mathrm{depth}][z]$
		                \State $\mathrm{sbox}[x] \gets y$
		                \State $b \gets\textsc{addPoint}(x) $
		                \If{$b$ is equal to $1$}
		                    \State $\textsc{nextVal}(\mathrm{depth}+1)$
		                \EndIf
		               \State $\mathrm{sbox}[x] \gets \bot$ 
		               \State $\textsc{removePoint}(x)$
		    \EndFor
		\EndFunction
		\vspace{1em}
		\Function{addPoint}{c}
		\If{$\textsc{addDDTInformation}(c)$}
		    \State \Return $\textsc{addDegreeInformation}(c)$
		\EndIf
		\State \Return 0
		\EndFunction
		
		\vspace{1em}
		\Function{removePoint}{c}
		\If{$\textsc{removeDDTInformation}(c)$}
		    \State $\textsc{removeDegreeInformation}(c)$
		\EndIf
		\EndFunction
	\end{algorithmic}
\end{algorithm}

\subsection{EA-equivalence Check}
For each function that we find, we need to check whether it is EA-equivalent to a known instance. We recall that for two quadratic APN functions, EA-equivalence coincides with CCZ-equivalence. To perform the check efficiently, we use the following method by Canteaut et al., first explained in an invited talk at Boolean Functions and their Applications (BFA) 2020 and formally described in the recent preprint~\cite{DBLP:journals/corr/abs-2103-00078}.

\begin{proposition}[\cite{ortho_derivative,DBLP:journals/corr/abs-2103-00078}]
Let $F \colon \F_2^n \rightarrow \F_2^n$ be a quadratic APN function. The \emph{ortho-derivative} of $F$ is defined as the unique function $\Pi_F \colon \F_2^n \rightarrow \F_2^n$ with $\Pi_F(0) = 0$ such that, for all $\alpha \in \F_2^n \setminus \{0\}$, we have $\Pi_F(\alpha) \neq 0$ and
\[ \forall x \in \F_2^n \colon \langle \Pi_F(\alpha), (F(x)+F(x+\alpha)+F(\alpha)+F(0)\rangle = 0.\]

For two EA-equivalent quadratic APN functions $F,G \colon \F_2^n \rightarrow \F_2^n$, the ortho-derivatives $\Pi_F$ and $\Pi_G$ are linear-equivalent.
\end{proposition}

Testing two quadratic APN functions for EA-inequivalence (which is the same as CCZ-in\-equiv\-a\-lence in this case) is now fairly simple. One simply computes the corresponding ortho-derivatives and evaluates their extended Walsh spectra and differential spectra. We note that these are (strongly discriminating) invariants for EA-equivalence. If the two extended Walsh spectra or the two differential spectra are different, the ortho-derivatives cannot be linear-equivalent and hence, the two quadratic APN functions cannot be EA-equivalent. The implementation for computing the ortho-derivative is contained in the latest version of sboxU~\cite{Perrin_sboxu}. This method is much more efficient than checking the code equivalence with Magma~\cite{MR1484478}.

 It might be the case that two EA-inequivalent functions are are not identified as such because their ortho-derivatives might have identical differential and extended Walsh spectra. However, this does not seem to occur often. For example, all of the previously-known 8,191 quadratic 8-bit APN instances can be established to be EA-inequivalent by this method. The complete check for all of those 8,191 APN instances only takes a few minutes on a PC.
 
 We remark that no further effort is needed to test the quadratic APN functions that we find for CCZ-inequivalence to the non-quadratic monomial APN functions. This is because of the fact that a quadratic APN function CCZ-equivalent to a monomial function must be EA-equivalent to a quadratic monomial function~\cite{yoshiara2016equivalences}.

\subsection{Results}
After running the search for $n=7$ for about 72 CPU hours, we found most of the quadratic APN instances, including the recently discovered APN function presented in~\cite{cryptoeprint:2020:1113}. For higher values of $n$, this direct approach is not very efficient and so we consider linear self-equivalences in the following.

\section{Considering LE-Automorphisms}
We now describe the method for searching for quadratic APN functions with non-trivial LE-automorphisms in small dimension $n$ and apply it to $n\in \{7,8,9,10\}$. Again, the algorithm is similar to the one presented in~\cite{beierle2020} where the focus was on APN permutations of arbitrary algebraic degree.
\subsection{Canonical classes of LE-Automorphisms}
If we consider $n$-bit functions $F$ with non-trivial elements in $\AutLE(F)$ and are only interested in a classification of such $F$ up to linear-equivalence,\footnote{We cannot use EA-equivalence here because the property of admitting a non-trivial LE-automorphism is not invariant under EA-equivalence.} we can significantly reduce the number of matrix pairs $(B,A)$ that we need to consider.  The following ideas and reductions have already been presented in~\cite{beierle2020} with a focus on the case of $F$ being a permutation.

\begin{lemma}[\cite{beierle2020}]
\label{lem:prime_order}
Let $F \colon \F_2^n \rightarrow \F_2^n$ have a non-trivial automorphism in $\AutLE(F)$. There exist $A,B \in \GL(n,\F_2)$ with $F \circ A= B \circ F$  such that either 
\begin{enumerate}
    \item $\ord(A)=\ord(B)=p$ for $p$ prime, or
    \item $B = I_n$ and $\ord(A) = p$ for $p$ prime, or
    \item $A = I_n$ and $\ord(B) = p$ for $p$ prime.
\end{enumerate}
\end{lemma}

Two elements $M,M' \in \GL(n,\F_2)$ are called \emph{similar}, denoted $M \sim M'$, if there exists  $P \in \GL(n,\F_2)$ such that $M' = P^{-1} M P$. Similarity is an equivalence relation on $\GL(n,\F_2)$ and we can find a representative of each equivalence class by the so-called rational canonical form.

\begin{lemma}{(Rational Canonical Form)\cite[Page~476]{dummit1991}}
\label{lem:rcf}
Every element $M \in \GL(n,\F_2)$ is similar to a unique  $M' \in \GL(n,\F_2)$ of the form $M' = \diag(\companion(q_r),\companion(q_{r-1}),\dots,\companion(q_1))$
for polynomials $q_i$ such that $q_r \mid q_{r-1} \mid \dots \mid q_1$. This matrix $M'$ is called the \emph{rational canonical form} of $M$, denoted $\rcf(M)$.
\end{lemma}

If we want to collect all LE-automorphisms in order to classify all functions $F \colon \F_2^n \rightarrow \F_2^n$ that admit a non-trivial linear self-equivalence up to linear-equivalence, we can without loss of generality assume that $A$ and $B$ are in rational canonical form. The following definition and lemma allow us to reduce the search space even further.
\begin{definition}[\cite{beierle2020}]
Let $A,B,C,D \in \GL(n,\F_2)$ be of order $p$ for $p$ prime. The tuple $(A,B)$ is said to be \emph{power-similar} to the tuple $(C,D)$, denoted $(A,B) \sim_p (C,D)$, if there exists a positive integer $i$ such that $A \sim C^i$ and $B \sim D^i$. 
\end{definition}
Power-similarity  defines an equivalence relation on the ordered pairs of matrices in $\GL(n,\F_2^n)$ of the same prime order and the following holds.

\begin{lemma}[\cite{beierle2020}]
\label{lem:power_similarity}
Let $F \colon \F_2^n \rightarrow \F_2^n$ with $\diag(A,B) \in \AutLE(F)$ for $A, B\in \GL(n,\F_2)$ being of prime order $p$. For every $(\widetilde{B},\widetilde{A})$ power-similar to $(B,A)$, there is a function $G$ linear-equivalent to $F$ such that $\diag(\rcf(\widetilde{A}),\rcf(\widetilde{B})) \in \AutLE(G)$. 
\end{lemma}

For the dimensions that we consider in this work, i.e., $n\leq 10$, we can efficiently generate all rational canonical forms as block-diagonal matrices of the form given in Lemma~\ref{lem:rcf}. By applying the above Lemmas~\ref{lem:prime_order} and~\ref{lem:power_similarity} to a fixed dimension $n$, we obtain a reduced number of tuples $(B,A)$ to consider. We call those \emph{canonical classes} of LE-automorphisms.

\paragraph{$n=7$}
All linear-equivalence classes of functions $F \colon \F_2^n \rightarrow \F_2^n$ that admit a non-trivial linear self-equivalence
can be obtained by considering 128 canonical classes for pairs $(B,A)$: 56 classes with $\ord(A)=\ord(B)$ being prime, 36 classes with $B = I_n$ and $\ord(A)$ being prime, and 36 classes with $A = I_n$ and $\ord(B)$ being prime. 

\paragraph{$n=8$}
All linear-equivalence classes of functions $F \colon \F_2^n \rightarrow \F_2^n$ that admit a non-trivial linear self-equivalence
can be obtained by considering 157 canonical classes for pairs $(B,A)$: 75 classes with $\ord(A)=\ord(B)$ being prime, 41 classes with $B = I_n$ and $\ord(A)$ being prime, and 41 classes with $A = I_n$ and $\ord(B)$ being prime. 

\paragraph{$n=9$}
All linear-equivalence classes of functions $F \colon \F_2^n \rightarrow \F_2^n$ that admit a non-trivial linear self-equivalence
can be obtained by considering 217 canonical classes for pairs $(B,A)$: 111 classes with $\ord(A)=\ord(B)$ being prime, 53 classes with $B = I_n$ and $\ord(A)$ being prime, and 53 classes with $A = I_n$ and $\ord(B)$ being prime. 

\paragraph{$n=10$}
All linear-equivalence classes of functions $F \colon \F_2^n \rightarrow \F_2^n$ that admit a non-trivial linear self-equivalence
can be obtained by considering 401 canonical classes for pairs $(B,A)$: 247 classes with $\ord(A)=\ord(B)$ being prime, 77 classes with $B = I_n$ and $\ord(A)$ being prime, and 77 classes with $A = I_n$ and $\ord(B)$ being prime.

The tuples $(B,A)$ and their corresponding class indices  can be found along with our source code at~\cite{our_code}.

\subsection{Finding quadratic APN functions in the canonical classes}
\label{sec:search_le_aut}

If we are searching for APN functions (not necessarily quadratic) with non-trivial LE-auto\-mor\-phisms, we do not have to check all of the canonical classes, as outlined in the following two lemmas.
For a matrix $M \in \GL(n,\F_2)$, let us denote by $\mathrm{Fix}_M \coloneqq \{ x \in \F_2^n \mid Mx = x\}$ the set of fixed points of $M$, which is a linear subspace of $\F_2^n$.

\begin{lemma}
\label{lem:admissible1}
Let $A,B \in \GL(n,\F_2)$. If $B$ has strictly less fixed points than $A$ and $(|\mathrm{Fix}_B|,|\mathrm{Fix}_A|) \notin \{(1,2),(2,4)\}$, no function $F\colon \F_2^n \rightarrow \F_2^n$ with $F \circ A = B \circ F$ can be APN. 
\end{lemma}
\begin{proof}
Let $F$ be such that $F \circ A = B \circ F$. For all $x \in \mathrm{Fix}_A$, we have $F(x) = B(F(x))$, thus $\{F(x) \mid x \in \mathrm{Fix}_A\} \subseteq \mathrm{Fix}_B$. Since $\mathrm{Fix}_A$ and $\mathrm{Fix}_B$ are linear subspaces of $\F_2^n$, if $|\mathrm{Fix}_B| < |\mathrm{Fix}_A|$, the image of the restriction of $F$ on the subspace $\mathrm{Fix}_A$ (with $\dim \mathrm{Fix}_A = k$) is contained in a subspace of smaller dimension $\ell < k$, where $\ell = \dim \mathrm{Fix}_B$. If $F$ is APN, this would imply the existence of an APN function $H$ on $\F_2^k$ whose image set is contained in $\F_2^\ell \times \{0\}^{k-\ell}$. The APN function $H$ must therefore have a component constant and equal to zero, more precisely, it must have $2^{k-\ell}-1$ of them. According to~\cite[Proposition 161]{carlet_2021}, this is not possible, unless $k \leq 2$. Since the cases of $(\ell,k) \in \{(0,1),(1,2)\}$ are excluded in the statement of the lemma, we only need to consider the case of $(\ell,k) = (0,2)$. In this case, we have that $H \colon \F_2^2 \rightarrow \F_2^2$ is constant and equal to zero, so $H$ is not an APN function.
\end{proof}

The following lemma is a direct consequence of~\cite[Lemma 5]{DBLP:conf/acns/CarletHP17} and states that classes of the form $(I_n,A)$ only need to be considered if the cycle decomposition of $A$ (where $A$ is considered as a permutation on the set $\F_2^n$) consists of a large number of cycles. 

\begin{lemma}
\label{lem:admissible2}
Let $A \in \GL(n,\F_2)$ and let $F \colon \F_2^n \rightarrow \F_2^n$ be an APN function with $F \circ A = F$. Let us denote by $c$ be the number of distinct cycles of $A$. If $n$ is even, we have  $c \geq \frac{2^n+2}{3}$, and, if $n$ is odd, we have $c \geq \frac{2^n+1}{3}$. 
\end{lemma}
\begin{proof}
For an $n$-bit function $F$ with $F \circ A = F$, we have for the image size $|\{F(x) \mid x \in \F_2^n\}| \leq c$. From~\cite[Lemma 5]{DBLP:conf/acns/CarletHP17} (see also~\cite{cryptoeprint:2020:1529} for a more recent discussion of this result), if $F$ is APN, we have
\[ c \geq \left\lceil \frac{2^{2n}}{3 \cdot 2^n -2}\right\rceil\;,\]
where the term on the right-hand side is equal to $\frac{2^n+2}{3}$ if $n$ is even and $\frac{2^n+1}{3}$ if $n$ is odd.
\end{proof}

For the remaining (non-trivial) canonical classes of LE-automorphisms in dimensions $n=7$ and $n=8$ not excluded by Lemmas~\ref{lem:admissible1} and~\ref{lem:admissible2} (i.e., 53 classes for $n=7$ and 67 classes for $n=8$) and for most of\footnote{We did not perform any search if $|\mathrm{Fix}_A| \geq 2^8$. For this reason, we excluded 2  and 7 of the remaining (non-trivial) canonical classes of LE-automorphisms in the search in dimension $n=9$ and $n=10$, respectively.} the remaining (non-trivial) canonical classes of LE-automorphisms in dimensions $n=9$ and $n=10$, we performed a randomized tree search for quadratic APN functions similar to the search described in Algorithm~\ref{alg:search}. The main difference is that, once we fix an element $F(x) = y$, we fix $F(A^i(x)) = B^i(y), i \in \{1,\dots,\ord_A(x)-1\}$ as well. Moreover, if $A$ and $B$ have an identical number of fixed points, we set the restriction of $F$ on $\mathrm{Fix}_A$ to an APN function before invoking the recursion. 
More precisely, let $F$ be a quadratic APN function with $F \circ A = B \circ F$, where $\dim \mathrm{Fix}_A = \dim \mathrm{Fix}_B = k$. Let further $\pi_A \colon \F_2^k \rightarrow \mathrm{Fix}_A$ and $\pi_B \colon \F_2^k \rightarrow \mathrm{Fix}_B$ be (vector space) isomorphisms. Then, there exists an APN function $G \colon \F_2^k \rightarrow \F_2^k$ of algebraic degree at most two such that, for all $x \in \F_2^k$, we have $F(\pi_A(x)) = \pi_B(G(x))$. It is not clear whether we can choose $G$ up to EA-equivalence. However, in our randomized search, we only select $G$ from a predetermined list of EA-representatives.
For example, if $k=6$, we choose one of 13 quadratic APN instances uniformly at random. We refer to our source code for the lists of EA-representatives from which we choose for each fixed $k$.
Note that our algorithm might terminate without returning a quadratic APN function although one exists. This is the case if the chosen $G$ does not yield a quadratic $n$-bit APN function, while another choice of $G$ would. However, we do not know whether this can happen.

\paragraph{Special cases} There are some special cases of classes for which we slightly modify the algorithm. Those cases are the LE-automorphisms for which $1 < |\mathrm{Fix}_A| < |\mathrm{Fix}_B|$. If $2 = |\mathrm{Fix}_A| < |\mathrm{Fix}_B|$ with $\mathrm{Fix}_A = \{0,x\}$, we set $F(x) = 0$ in the beginning before calling $\textsc{NextVal}(0)$. However, if  $2 < |\mathrm{Fix}_A| < |\mathrm{Fix}_B|$, we do not set any points in the beginning besides $F(0)=0$. Instead, we first select those positions $x$ for which $\ord_A(x) >1$  in $\textsc{NextFreePosition}()$ and the fixed points of $A$ last. This allows us to fix large cycles first. 

\paragraph{Deterministic search} It is possible to slightly modify the randomized algorithm to perform a deterministic search. In particular, we run the search for all possible choices of $k$-bit APN instances $G$ of algebraic degree at most two and do not abort after a predetermined amount of time or if a solution is found. Similarly as it was done in the exhaustive search presented in~\cite{beierle2020}, we call $\textsc{NextVal}(\mathrm{depth}+1)$ only if $F$ is the smallest representative (up to some previously-defined ordering) in the set $\{ C_B \circ F \circ C_A \mid C_A \in \mathcal{C}_A, C_B \in \mathcal{C}_B\}$, where $\mathcal{C}_A$ (resp., $\mathcal{C}_B$) is a subset of all elements in $\GL(n,\F_2)$ that commute with $A$ (resp., $B$) and whose restriction on $\mathrm{Fix}_A$ (resp., $\mathrm{Fix}_B$) is the identity. Formally, $\mathcal{C}_A \subseteq \{ M \in \GL(n,\F_2) \mid M A = A M \text{ and } \forall x \in \mathrm{Fix}_A: Mx = x\}$ (similarly for $\mathcal{C}_B$). Note that this method does not necessarily result in an exhaustive search as it is not clear whether we can select $G$ up to EA-equivalence. 

For each canonical class of LE-automorphisms, before starting the randomized search, we first check whether the deterministic search terminates in short time (i.e., a few minutes to hours). If this is the case, we do not invoke the randomized search. 

\paragraph{Results for $n=7$} The only APN functions found are those which are EA-equivalent to univariate polynomials with coefficients in $\F_2$. 
Quadratic APN polynomials with coefficients in $\F_2$ have already been classified for $n=7$, see~\cite{cryptoeprint:2019:1491}. 

\paragraph{Results for $n=8$}
To the best of our knowledge, by the time of submission of this manuscript in December 2020 and excluding our results, there are 8,192 known instances of 8-bit APN functions, i.e., the 23 instances listed in~\cite{DBLP:journals/amco/EdelP09}, the 8,157 instances constructed by the QAM method~\cite{DBLP:journals/iacr/YuWL13, DBLP:journals/dcc/YuWL14}, the 10 instances presented in~\cite{weng2013quadratic}, and the two instances from the Taniguchi family~\cite{DBLP:journals/dcc/Taniguchi19a}.
With our approach, we find 12,733 new instances of quadratic APN functions. In particular, we find quadratic APN functions within 9 different classes of LE-automorphisms (see Table~\ref{tab:n8}). For each canonical class of LE-automorphisms that could not be directly excluded by Lemma~\ref{lem:admissible1} or Lemma~\ref{lem:admissible2}, we performed the search for a few CPU days at most.\footnote{Note that for a lot of classes, the search terminates immediately without solutions.} 

The extended Walsh spectra of all of those 12,733 new instances belong to one of the six spectra $\mathcal{W}_0,\dots,\mathcal{W}_5$ listed below. By $a: m$, we indicate that the value $a$ occurs with multiplicity $m$ in the multiset. APN functions in dimension $n=8$ having $\mathcal{W}_0, \mathcal{W}_1$, or $\mathcal{W}_2$ as their extended Walsh spectrum are already known. There was no previously-known APN function with extended Walsh spectrum $\mathcal{W}_3, \mathcal{W}_4$, or $\mathcal{W}_5$. 
\begin{align*}
    \mathcal{W}_0 &= \{0: 16320, 16: 43520, 32: 5440\} \quad \text{ (classical spectrum)} \\
    \mathcal{W}_1 &= \{0: 15600, 16: 44544, 32: 5120, 64: 16\} \\
    \mathcal{W}_2 &= \{0: 14880, 16: 45568, 32: 4800, 64: 32\} \\ 
    \mathcal{W}_3 &= \{0: 14160, 16: 46592, 32: 4480, 64: 48\} \quad \text{(new)} \\
    \mathcal{W}_4 &= \{0: 13440, 16: 47616, 32: 4160, 64: 64\}  \quad \text{(new)}\\
    \mathcal{W}_5 &= \{0: 12540, 16: 48640, 32: 4096, 128: 4\}  \quad \text{(new)}\\
\end{align*}

Certainly, $\mathcal{W}_5$ is the most interesting extended Walsh spectrum since it corresponds to 8-bit functions with linearity $2^{7}$ and we found four such instances of quadratic APN functions. Whether quadratic $n$-bit APN functions with linearity $2^{n-1}$ exist was mentioned as an open problem in~\cite{DBLP:journals/tit/Carlet18}. Before now, besides the trivial cases in dimension $n \leq 4$, we only knew one such instance in dimension $n=6$, see~\cite{DBLP:journals/amco/EdelP09}. 

\begin{table}[h!]
\caption{\label{tab:n8}For all class indices (no.) in dimension 8 for which we find solutions, this table gives a lower bound on the number of distinct EA-equivalence classes of quadratic APN functions that admit the particular LE-automorphism,  separated by their extended Walsh spectra. The numbers in parentheses indicate the number of  instances that are not contained in the previously known 8,192 instances of APN functions. }
\centering
\begin{tabular}{ccccccc}
    \toprule
     no.  & $\mathcal{W}_0$ & $\mathcal{W}_1$ & $\mathcal{W}_2$ & $\mathcal{W}_3$ & $\mathcal{W}_4$ & $\mathcal{W}_5$   \\
     \midrule
     1 & 1 (0) & \\
     2 & 1 (0) & \\
     21 & 9 (9) & \\
     31 & 7 (4) & 1 (1) & & 2 (2) \\
     38 & 3 (0) & \\
     51 & 24 (20) & \\
     55 & 9,093 (9,090) & 3,065 (3,065) & 299 (297) & 146 (146) & 25 (25) & 4 (4) \\
     56 & 103 (79) & & 2 (2) & & 1 (1) & \\
     113 & 26 (0) & \\
     \bottomrule
     \end{tabular}
\end{table}

In Class 1 (up to a change of basis), the matrix $A$ corresponds to multiplication by a non-zero field element $\alpha \in \F_{2^n}$ of multiplicative order 17 and the matrix $B$ corresponds to multiplication by $\alpha^3$. Thus, as a solution, we find the APN function $x \mapsto x^3$. Similarly, in Class 2, the matrix $A$ corresponds to multiplication by a non-zero field element $\alpha \in \F_{2^n}$ of multiplicative order 17 and the matrix $B$ corresponds to multiplication by $\alpha^9$. Therefore, as a solution, we find the APN function $x \mapsto x
^9$.

Class 56 corresponds to those functions $F$ whose univariate representation only contains coefficients in the subfield $\F_{2^4}$.

Let $\zeta_3 \in \F_{2^8}^* \setminus \{1\}$ be a third root of unity. Class 113 corresponds to those functions $F$ for which $F(x) = F(\zeta_3 x)$. Class 51 corresponds to the linear self-equivalence where $A$ is the multiplication by $\zeta_3$ and $B$ is similar to the block diagonal matrix $\diag(I_2,\companion(X^3+1),\companion(X^3+1))$
and Class 55 corresponds to the linear self-equivalence where $A$ is the  multiplication by $\zeta_3$ and $B$ is similar to the block-diagonal matrix $\diag(I_5,\companion(X^3+1))$.
 Note that in both cases, $B$ does not correspond to multiplication by a finite field element or to a linear mapping of the form $x \mapsto x^{2^i}$. 

Let $\zeta_5 \in \F_{2^8}^*\setminus \{1\}$ be a fifth root of unity. Class 38 corresponds to those functions $F$ for which $\zeta_5(F(x)) = F(\zeta_5 x)$. Class 31 corresponds to the linear self-equivalence where $A$ is the multiplication by $\zeta_5$ and $B$ is similar to the block-diagonal matrix $\diag(I_3,\companion(X^5+1))$. 
Again, $B$ does not correspond to multiplication by a finite field element or to a linear mapping of the form $x \mapsto x^{2^i}$.

Class 21 corresponds to the linear self-equivalence given by $A=B=\diag(I_1,\companion(X^7+1))$.

\paragraph{Results for $n=9$}
To the best of our knowledge, by the time of submission, the only known APN instances for $n=9$ either correspond to polynomials with coefficients in $\F_{2}$~\cite{cryptoeprint:2019:1491}, to the (generalized) isotopic shift construction~\cite{DBLP:journals/tit/BudaghyanCCCV20,budaghyan2020generalized}, or to the infinite families given in~\cite{DBLP:journals/ffa/BudaghyanCL09,budaghyan2009construction}. The only APN \emph{permutations} in dimension nine known so far are CCZ-equivalent to monomial functions. We applied our search for $n=9$ and found 35 new APN instances, two of them being permutations. Up to linear-equivalence, those two instances of APN permutations $F_1,F_2$  admit a linear self-equivalence of the form $F_i(u^5x) = uF_i(x)$, where $u \in \F_{2^3}^*$. 

Functions fulfilling this self-equivalence can be characterized by the property that their univariate representation does not contain monomials $x^j$ with $j \neq 3 \mod 7$. Interestingly, this property is also fulfilled by the ``Kim mapping'', i.e., the function
\[K \colon \F_{2^6} \rightarrow \F_{2^6}, \quad x \mapsto x^3 + x^{10} + gx^{24},\]
where $g$ is an element in $\F_{2^6}^*$ with minimal polynomial $X^6 + X^4 + X^3 + X + 1 \in \F_2[X]$. It was shown in~\cite{browning2010apn} that $K$ is CCZ-equivalent to an APN permutation.

\paragraph{Results for $n=10$}
To the best of our knowledge, by the time of submission, the only known APN instances for $n=10$ are either monomial functions or those that come from the infinite families given in~\cite{DBLP:journals/tit/BudaghyanC08}, \cite{DBLP:journals/ffa/BudaghyanCL09}, \cite{DBLP:journals/dcc/Taniguchi19a}, and \cite{DBLP:journals/iacr/BudaghyanHK19} (see the instances 10.1--10.17 in the list available at \url{https://boolean.h.uib.no/mediawiki/index.php/CCZ-inequivalent_representatives_from_the_known_APN_families_for_dimensions_up_to_11}).\footnote{accessed October 6, 2021. As noted in the disclaimer, not all instances of the family given in~\cite{DBLP:journals/tit/BudaghyanCCCV20} could have been checked. Therefore, it could be possible that our instances are coming from that family.} We applied our approach for $n=10$ and found 5 APN instances that are CCZ-inequivalent to the known 17 instances that come from infinite families. 

The look-up tables of all of the new APN instances that we found in dimension $n\in\{8,9,10\}$ are available in~\cite{dataset}.

\section{Further APN Instances from the Switching Construction}
In~\cite{DBLP:journals/amco/EdelP09}, the authors presented the switching construction which potentially allows to generate CCZ-inequivalent APN functions from a given APN function by replacing one of its components.  Using this method, they found the only known APN instance that is not CCZ-equivalent to a monomial function or to a quadratic function. 

\begin{definition}
Two functions $F,G \colon \F_2^n \rightarrow \F_2^n$ are said to be \emph{switching neighbours} if there exists a Boolean function $f \colon \F_2^n \rightarrow \F_2$ and a non-zero vector $v \in \F_2^n$ such that, for all $x \in \F_2^n$, we have $G(x) = F(x) + vf(x)$. 
\end{definition}

It has been shown in~\cite[Theorem 3]{DBLP:journals/amco/EdelP09} that, if $F\colon \F_2^n \rightarrow \F_2^n$ is an APN function, all APN functions that are switching neighbours of $F$ can be found by simple linear algebra. Indeed, for a non-zero vector $v \in \F_2^n$, to find all functions $f \colon \F_2^n \rightarrow \F_2$ such that $F + vf$ is an APN function, one can collect all tuples $(x,y,x+y,a), x,y,a \in \F_2^n$ for which $F(x)+F(x+a)+F(y)+F(y+a)=v$ holds and for each such tuple include the linear equation
\[f_x+f_{x+a}+f_y+f_{y+a} = 0 \]
to the system to solve, where the unknowns are $f_z \in \F_2$ for $z \in \F_2^n$.
Any solution $(f_z)_{z \in \F_2^n}$ of this linear system corresponds to a function $f$ with $f(z) = f_z, z \in \F_2^n$ such that $F + vf$ is APN. 

We applied this method to all of the known (by the time of submission of this manuscript) and new APN instances for $n=7$ and $n=8$. For $n=8$, we found 188 new APN instances in this way. All of those functions are quadratic. The look-up tables of those 188 functions are also available in the dataset~\cite{dataset}.

\section{Conclusion}
We performed a recursive search for quadratic APN functions in small dimension and we have shown that quadratic APN functions with linearity $2^{n-1}$ exist for the case of $n=8$. Further, we found two previously unknown APN permutations in dimension $n=9$. An open question is whether APN functions with linearity $2^{n-1}$ also exist in any even dimension $n > 8$. To answer this question, it would be interesting to generalize one of the four eight-bit APN functions with high linearity presented in this work to an infinite family. Another question is whether the list of known extended Walsh spectra of quadratic APN functions in dimension $n=8$ is complete. Finally, the problem of classifying the new APN permutations in dimension $n=9$ into infinite families remains open. In this direction, it would be interesting to analyze similarities between the Kim function and the two new APN permutations.

\section*{Acknowledgment}
We thank the associate editor and the anonymous reviewers for their detailed and
helpful comments. We further thank Léo Perrin for pointing us to the idea of using the ortho-derivative for checking EA-equivalence of our found functions.

\newcommand{\etalchar}[1]{$^{#1}$}

\end{document}